\documentclass[draft,a4paper]{llncs}
\usepackage[utf8x]{inputenc}
\usepackage{amstext, amsmath, amssymb,amscd}
\usepackage{hyperref}
\newtheorem{observation}{Observation}

\usepackage{xspace}

\newcommand{\F}{\ensuremath{\mathcal{F}}\xspace}
\newcommand{\N}{\mathbb{N}}
\newcommand{\Oh}{\mathcal{O}}

\newcommand{\yes}{\textsc{yes}\xspace}
\newcommand{\no}{\textsc{no}\xspace}

\newcommand{\NP}{\ensuremath{\mathsf{NP}}\xspace}

\renewcommand{\P}{\ensuremath{\mathsf{P}}\xspace}

\newcommand{\halt}{\mathbin{\rangle\mkern-4mu|}}
\newcommand{\resume}{\mathbin{|\mkern-4mu\langle}}

\newcommand{\dHSk}{$d$-\textsc{Hitting Set($k$)}\xspace}
\newcommand{\dSMk}{$d$-\textsc{Set Matching($k$)}\xspace}
\newcommand{\EDSk}{\textsc{Edge Dominating Set($k$)}\xspace}
\newcommand{\CEk}{\textsc{Cluster Editing($k$)}\xspace}

\newcommand{\MFIk}{\textsc{Minimum Fill-In($k$)}\xspace}

\newcommand{\FBVSk}{\textsc{Feedback Vertex Set($k$)}\xspace}
\newcommand{\OCTk}{\textsc{Odd Cycle Transversal($k$)}\xspace}
\newcommand{\CVDk}{\textsc{Cluster Vertex Deletion($k$)}\xspace}
\newcommand{\CEDk}{\textsc{Cluster Deletion($k$)}\xspace}
\newcommand{\BERk}{\textsc{Edge Bipartization($k$)}\xspace}

\newcommand{\BCNk}{\textsc{Bipartite Colorful Neighborhood($k$)}\xspace}

\newcommand{\CoVDk}{\textsc{Cograph Vertex Deletion($k$)}\xspace}
\newcommand{\TEDk}{\textsc{Triangle Edge Deletion($k$)}\xspace}
\newcommand{\TVDk}{\textsc{Triangle Vertex Deletion($k$)}\xspace}
\newcommand{\TPk}{\textsc{Triangle Packing($k$)}\xspace}
\newcommand{\sSPk}{$s$-\textsc{Star Packing($k$)}\xspace}

\newcommand{\introduceparameterizedproblem}[4]{
\fbox{
\parbox{0.96\textwidth}{
	#1\hfill \textbf{Parameter:} #3\\
	\textbf{Input:} #2\\	
	\textbf{Question:} #4}}
\vspace{0.3cm}	
}

\title{\texorpdfstring{Streaming Kernelization\thanks{Supported by the Emmy Noether-program of the DFG, KR 4286/1.}}{Streaming Kernelization}} 
\author{Stefan Fafianie \and Stefan Kratsch}
\institute{TU Berlin, Germany, \texttt{$\{$stefan.fafianie,stefan.kratsch$\}$@tu-berlin.de}}

\begin{document}
\maketitle

\begin{abstract}
Kernelization is a formalization of preprocessing for combinatorially hard problems. We modify the standard definition for kernelization, which allows any polynomial-time algorithm for the preprocessing, by requiring instead that the preprocessing runs in a streaming setting and uses $\Oh(poly(k)\log|x|)$ bits of memory on instances $(x,k)$. We obtain several results in this new setting, depending on the number of passes over the input that such a streaming kernelization is allowed to make. \textsc{Edge Dominating Set} turns out as an interesting example because it has no single-pass kernelization but two passes over the input suffice to match the bounds of the best standard kernelization. 

\end{abstract}

\section{Introduction}

When faced with an \NP-hard problem we do not expect to find an efficient algorithm that solves every instance exactly and in polynomial time (as this would imply \P~$=$~\NP). The study of algorithmic techniques offers various paradigms for coping with this situation if we are willing to compromise on efficiency, exactness, or the generality of being applicable to all instances (or several of those). Before we commit to such a compromise it is natural to see how much closer we can come to a solution by spending only polynomial time, i.e., how much we can simplify and shrink the instance by polynomial-time \emph{preprocessing}. This is usually compatible with any way of solving the simplified instance and it finds wide application in practice (e.g., as a part of ILP solvers like CPLEX), although, typically, the applications are of a heuristic flavor with no guarantees for the size of the simplified instance or the amount of simplification.

The notion of \emph{kernelization} is one way of formally capturing preprocessing. A kernelization algorithm applied to some problem instance takes polynomial time in the input size and always returns an equivalent instance (i.e., the instances will have the same answer) of size bounded by a function of some \emph{problem-specific parameter}. For example, the problem of testing whether a given graph~$G$ has a vertex cover of size at most~$k$ can be efficiently reduced to an equivalent instance~$(G',k)$ where~$G'$ has~$\Oh(k)$ vertices and~$\Oh(k^2)$ total bit size. The study of kernelization is a vibrant field that has seen a wealth of new techniques and results over the last decade. (The interested reader is referred to recent surveys by Lokshtanov et al.~\cite{LokshtanovMS12} and Misra et al.~\cite{MisraRS11}.) In particular, a wide-range of problems is already classified into admitting or not admitting\footnote{Unless $\mathsf{NP\subseteq coNP/poly}$ and the polynomial hierarchy collapses.} a \emph{
polynomial} kernelization, where the guaranteed output size bound is polynomial in the chosen parameter.
It is seems fair to say that this shows a substantial \emph{theoretical} success of the notion of kernelization.

From a practical point of view, we might have to do more work to convince a practitioner that our positive kernelization results are also \emph{worth implementing}. This includes choice of parameter, computational complexity, and also conceptual difficulty (e.g., number of black box subroutines, huge hidden constants). Stronger parameterizations already receive substantial interest from a theoretical point of view, see e.g.,~\cite{DBLP:journals/ejc/FellowsJR13}, and there is considerable interest in making kernelizations fast, see e.g.,~\cite{BevernHKNW11,Hagerup11,Bevern12,Kammer13}. Conceptual difficulty is of course ``in the eye of the beholder'' and perhaps hard to quantify.

In this work, we take the perspective that kernelizations that work in a restricted model might, depending on the model, be provably robust and useful/implementable (and hopefully also fast). Concretely, in the spirit of studying restricted models, we ask which kernelizations can be made to work in a \emph{streaming model} where the kernelization has a small local memory and only gets to look at the input once, or a bounded number of times. The idea is that the kernelization should maintain a sufficiently good sketch of the input that in the end will be returned as the reduced instance.

We think that this restricted model for kernelization has several further benefits: First of all, it naturally improves the memory access patterns since the input is read sequentially, which should be beneficial already for medium size inputs. (It also works more naturally for huge inputs, but huge instances of \NP-hard problems are probably only really addressable by outright use of heuristics or sampling methods.) Second, it is naturally connected to a dynamic/incremental setting since, due to the streaming setting, the algorithm has not much choice but to essentially maintain a simplified instance of bounded size that is equivalent to the input seen so far (or be able to quickly produce one should the end of the stream be declared). Thus, as further input arrives, the problem kernel is adapted to the now slightly larger instance without having to look at the whole instance again. (In a sense, the kernelization could run in parallel to the \emph{creation} of the 
actual input.) Third, it 
appears, at least in our positive results, that one could easily translate this to a parallel setting where, effectively, several copies of the algorithm work on different positions on the stream to simplify the instance (this however would require that an algorithm may delete data from the stream).

\emph{Our results.}
In this work we consider a streaming model where elements of a problem instance are presented to a kernelization algorithm in arbitrary order. The algorithm is required to return an equivalent instance of size polynomial in parameter $k$ after the stream has been processed. Furthermore, it is allowed to use $\Oh(poly(k) \log n)$ bits of memory, i.e., an overhead factor of $\Oh(\log n)$ is used in order to distinguish between elements of an instance of size $n$.

We show that \dHSk and \dSMk admit streaming kernels of size $\Oh(k^d \log d)$ while using $\Oh(k^d \log |U|)$ bits of memory where $U$ is the universal set of an input instance. We then consider a single pass kernel for \EDSk and find that it requires at least $m-1$ bits of memory for instances with $m$ edges. This rules out streaming kernels with $c \cdot poly(k) \log n$ bits for instances with $n$ vertices since for any fixed $c$ and $poly(k)$ there exist instances with $m-1 > c \cdot poly(k) \log n$. Insights obtained from this lower bound allow us to develop a general lower bound for the space complexity of single pass kernels for a class of parameterized graph problems. 

Despite the lower bound for single pass kernels, we show that \EDSk admits a streaming kernel if it is allowed to make a pass over the input stream twice. Finally, we use communication complexity games in order to rule out similar results for \CEk and \MFIk and show that multi-pass streaming kernels for these problems must use $\Omega(n)$ bits of local memory for graphs with $n$ vertices, even when a constant number of passes are allowed. 

\emph{Related work.} The data stream model is formalized by Henzinger et al.~\cite{raghavan1999computing}. Lower bounds for exact and randomized algorithms with a bounded number of passes over the input stream are given for various graph problems and are proven by means of communication complexity. An overview is given by Babcock et al.~\cite{babcock2002models} in which issues that arise in the data stream model are explored. An introduction and overview of algorithms and applications for data streams is given by Muthukrishnan \cite{muthukrishnan2005data}.

\emph{Organization.} 
Section \ref{sec:prelim} contains preliminaries and a formalization of kernelization algorithms in the data streaming setting. Single pass kernels for \dHSk and \dSMk are presented in Section \ref{sec:singlepass}. The lower bounds for single pass kernels are given in Section \ref{sec:singlepass:bounds}. The 2-pass kernel for \EDSk is shown in Section \ref{sec:2pass} while lower bounds for multi-pass kernels are given in Section \ref{sec:multi}. Finally, Section \ref{sec:conc} contains concluding remarks.

\section{Preliminaries} \label{sec:prelim}
We use standard notation from graph theory. For a set of edges $E$, let $V(E)$ be the set of vertices that are 
incident with edges in $E$. For a graph $G=(V,E)$, let $G[V]$ denote the subgraph of $G$ induced by $V$.  Furthermore, let $G[E]$ be the subgraph induced by $E$, i.e. $G[E] = G(V(E), E)$. 

A \emph{parameterized problem} is a language $Q \subseteq \Sigma^* \times \N$;
the second component~$k$ of instances~$(x,k)$ is called the \emph{parameter}. 
A parameterized problem is \emph{fixed-parameter tractable} if there is an algorithm that decides if $(x, k) \in Q$ 
in $f(k)|x|^{\mathcal{O}(1)}$ time, where $f$ is any computable function. 
A \emph{kernelization algorithm (kernel)} for a parameterized problem $Q \subseteq \Sigma^* \times N$ is an algorithm that, for input $(x, k) \in \Sigma^* \times \N$ outputs a pair $(x', k') \in \Sigma^* \times \N$ in $(|x|+k)^{\mathcal{O}(1)}$ time such that $|x'|, k' < g(k)$ for some computable function $g$, called the \emph{size} of the kernel, and $(x, k) \in Q \Leftrightarrow (x', k') \in Q$. A \emph{polynomial kernel} is a kernel with polynomial size.

\subsection*{Kernelization in the data-streaming model}

An input stream is a sequence of elements of the input problem. 
We denote the start of an input stream by $\langle$ and let $\rangle$ denote the end, 
e.g. $\langle e_1, e_2, \ldots, e_m \rangle$ denotes an input stream for a sequence of $m$ elements. 
We use $\halt$ to denote a halt in the stream and $\resume$ to denote its continuation, 
e.g. $\langle e_1, e_2 \halt$ and $\resume e_3, \ldots, e_m \rangle$ denote the same input stream broken up in two parts.

A \emph{streaming kernelization algorithm} (\emph{streaming kernel}) is an algorithm that receives input $(x, k)$ for a
parameterized problem in the following fashion. The algorithm is presented with an input stream where elements of $x$ are presented in a sequence,
i.e. adhering to the cash register model \cite{muthukrishnan2005data}. 
Finally, the algorithm should return a kernel for the problem upon request. A \emph{$t$-pass streaming kernel} is a streaming kernel that is allowed $t$ passes over the input stream before a kernel is requested. 

If $x$ is a graph, then the sequence of elements of $x$ are its edges in arbitrary ordering. In a natural extension to hypergraphs, if $x$ is a family of subsets on some ground set $U$, then the sequence of elements of $x$ are the sets of this family in arbitrary ordering. We assume that a streaming kernelization algorithm receives parameter $k$ and the size of the vertex set (resp. ground set) before the input stream. Note that this way isolated vertices are given implicitly.

Furthermore, we require that the algorithm uses a limited amount of space at any time during its execution. In the strict streaming kernelization setting the streaming kernel must use at most $p(k)\log |x|$ space where $p$ is a polynomial.
We will refer to a 1-pass streaming kernelization algorithm which upholds these space bounds simply as a \emph{streaming kernelization}.

We assume that words of size $\log |x|$ in memory can be compared in $\Oh(1)$ operations when considering the running time of the streaming kernelization algorithms in each step.

\section{Single pass kernelization algorithms}\label{sec:singlepass}
In this section we will show streaming kernelization algorithms for $d$-\textsc{Hitting Set}$(k)$ 
and $d$-\textsc{Set Matching}$(k)$ in the 1-pass data-stream model. 
These algorithms make a single pass over the input stream after which they output a kernel.
We analyze their efficiency with regard to local space and the worst case processing time for a 
single element in the input stream.

\medskip
\noindent
\introduceparameterizedproblem{\dHSk}{A set~$U$ and a family~$\F$ of subsets of~$U$ each of size at most~$d$, i.e.~$\F\subseteq\binom{U}{\leq d}$, and~$k\in \N$.}{$k$.}{Is there a set~$S$ of at most~$k$ elements of~$U$ that has a nonempty intersection with each set in~$\F$?}

In the following, we describe a single step of the streaming kernelization. After Step~$t$, the algorithm has seen a set~$\F_t\subseteq \F$, where~$\F$ denotes the whole set of edges provided in the stream. The memory contains some subset~$\F'_t\subseteq \F_t$, using for each~$F\in\F'_t$ a total of at most~$d\log n=\Oh(\log n)$ bits to denote the up to~$d$ elements therein. The algorithm maintains the invariant that the number of sets~$F\in\F'_t$ that contain any~$C\in\binom{U}{\leq d-1}$ as a subset is at most $(d-|C|)!\cdot (k+1)^{d-|C|}$. For intuition, let us remark that this strongly relates to the sunflower lemma. Now, let us consider Step~$t+1$. The memory contains some~$\F'_t\subseteq \F_t$ and a new set~$F$ arrives.
\smallskip

1. Iterate over all subsets~$C$ of~$F$, ordered by decreasing size.

2. Count the number of sets in~$\F'_t$ that contain~$C$ as a subset.

3. If the result equals~$(d-|C|)!\cdot (k+1)^{d-|C|}$ then the algorithm decides not to store~$F$ and ends the computation for Step~$t+1$, i.e., let~$\F'_{t+1}=\F'_t$.
 
4. Else, continue with the next set~$C$.
 
5. If no set~$C\subseteq F$ gave a total of~$(d-|C|)!\cdot (k+1)^{d-|C|}$ sets containing~$F$ then the algorithm decides to store~$F$, i.e.,~$\F'_{t+1}=\F'_t\cup\{F\}$. Note that this preserves the invariant for all~$C\in\binom{U}{d-1}$ since only the counts for~$C$ with~$C\subseteq F$ can increase, but all those were seen to be strictly below the threshold~$(d-|C|)!\cdot (k+1)^{d-|C|}$ so they can at most reach equality by adding~$F$.

\smallskip
To avoid confusion, let us point out that at any time the algorithm only has a single set~$\F'_t$; the index~$t$ is used for easier discussion of the changes over time.

\begin{observation}\label{obs:dhs}
The algorithm stores at most~$d!(k+1)^d=\Oh(k^d)$ sets at any point during the computation. This follows directly from the invariant when considering~$C=\emptyset$.
\end{observation}

\begin{theorem}\label{thm:dhs} \emph{($\bigstar$\footnote{Proofs of statements marked with $\bigstar$ are postponed to the appendix.})}
\dHSk admits a streaming kernelization which, using $\Oh(k^d \log |U|)$ bits of local memory and $\Oh(k^d)$ time in each step, returns an equivalent instance of size $\Oh(k^d \log k)$.
\end{theorem}

The time spent in each step can be improved from $\Oh(|\F'_t|)$ to $\Oh(\log |\F'_t|)$ at the cost of an increase in local space
by a constant factor. This can be realized with a tree structure $\mathbb{T}$ in which the algorithm maintains the number of sets in $\F'_t$ that contain a set $C \in \binom{U}{\leq d-1}$ as a subset.

Each $C \subseteq F'$, $F' \in \F'_t$ has a corresponding node in $\mathbb{T}$ and in this node the number of supersets of $C$ in $\F'_t$ are stored. Let the root node represent $C = \emptyset$ with a child for each set $C$ of size 1. In general, a node is assigned an element in $e \in \bigcup \F'_t$ and represents $C = C' \cup \{e\}$ where $C'$ is the set represented by its parent, i.e. $|C| = d$ for nodes with depth $d$. 

For each node, let $e_i$ be assigned to child node $n_i$. Furthermore, each node has a dictionary, i.e. a collection of (key, value) pairs $(e_i, n_i)$ in order to facilitate quick lookup of its children. Let us assume that there is some arbitrary ordering on elements that are in sets of $\F'_t$, e.g. by their identifier. Then the dictionary can be implemented as a self-balancing binary search tree. This allows us to find a child node and insert new child node in time $\Oh(\log h)$ if there are $h$ children.

\begin{corollary}\label{cor:dhs} \emph{($\bigstar$)}
 \dHSk admits a streaming kernelization which, using $\Oh(k^d \log |U|)$ bits of local memory and $\Oh(\log k)$ time in each step, returns an equivalent instance of size $\Oh(k^d \log k)$.
\end{corollary}

\medskip
\noindent
\introduceparameterizedproblem{\dSMk}{A set~$U$ and a family~$\F$ of subsets of~$U$ each of size at most~$d$, i.e.~$\F\subseteq\binom{U}{\leq d}$, and~$k\in \N$.}{$k$.}{Is there a matching $M$ of at least~$k$ sets in~$\F$, i.e. are there~$k$ sets in~$\F$ that are pairwise disjoint?}

The streaming kernelization will mostly perform the same operations in a single step as the algorithm described above such that only the invariant differs. In this case it is maintained that the number of sets $F \in \F'_t$ that contain any $C \in \binom{U}{\leq d-1}$ as a subset is at most $(d-|C|)!\cdot(d(k-1)+1)^{d-|C|}$.

\begin{observation}
 The algorithm stores at most~$d!(d(k-1)+1)^d = \Oh(k^d)$ sets at any point during the computation. This follows directly from the invariant when considering $C = \emptyset$.
\end{observation}

\begin{theorem} \label{thm:dsm}  \emph{($\bigstar$)}
 \dSMk admits a streaming kernelization which, using $\Oh(k^d \log |U|)$ bits of local memory and $\Oh(k^d)$ time in each step, returns an equivalent instance of size $\Oh(k^d \log k)$.
\end{theorem}

Similar to the algorithm described in the previous section, the running time in each step can be improved at the cost of an increase in local space by a constant factor. We omit an explicit proof.

\begin{corollary}
 \dSMk admits a streaming kernelization which, using $\Oh(k^d \log |U|)$ bits of local memory and $\Oh(\log k)$ time in each step, returns an equivalent instance of size $\Oh(k^d \log k)$.
\end{corollary}

\section{Space lower bounds for single pass kernels} \label{sec:singlepass:bounds}

We will now present lower bounds on the memory requirements of single pass streaming kernelization algorithms for a variety of graph problems. Before giving a general lower bound we first illustrate the essential obstacle by considering the \EDSk problem.  We show that a single pass kernel for \EDSk requires at least $m-1$ bits of memory on instances with~$m$ edges.

\medskip\noindent
\introduceparameterizedproblem{\EDSk}{A graph $G = (V,E)$ and $k \in \N$.}{$k$.}{Is there a set~$S$ of at most~$k$ edges such that every edge in~$E \setminus S$ is incident with an edge in~$S$?}

An obstacle that arises for many problems, such as \EDSk, is that they are not monotone under adding additional edges, i.e., additional edges do not always increase the cost of a minimum edge dominating set but may also decrease it. This decrease, however, may in turn depend on the existence of a particular edge in the input. Thus, on an intuitive level, it may be impossible for a streaming kernelization to ``decide'' which edges to forget, since worst-case analysis effectively makes additional edges behave adversarial. (Note that our lower bound does not depend on assumptions on what the kernelization decides to store.)

Consider the following type of instance as a concrete example of this issue. The input stream contains the number of vertices (immaterial for the example), the parameter value~$k=1$, and a sequence of edges $\langle \{a, v_1\} \ldots \{a, v_{n}\}, \{b, v\} \rangle$. That is, the first~$n$ edges form a star with $n$ leaves and center vertex~$a$. In order to use a relatively small amount of local memory the kernelization algorithm is forced to do some compression such that not every edge belonging to this star is stored in local memory. Now a final edge arrives and the algorithm returns a kernel. Note that the status of the problem instance depends on whether or not this edge is disjoint from the star: If it shares at least one vertex~$v_i$ with the star then there is an edge dominating set~$\{a,v_i\}$ of size one. Otherwise, if it is disjoint then clearly at least two edges are needed. Thus, from the memory state after the final edge we must be able to extract whether or not~$v$ is contained in~$\{v_1,\ldots,
v_n\}$; in other words, this is equivalent to whether or not the output kernelized instance is \yes or \no. (We assume that~$a,b\notin\{v_1,\ldots,v_n\}$ for this example.) This, however, is a classic problem for streaming algorithms that is related to the \emph{set 
reconciliation problem} and it is known to require at least~$n$ bits \cite{muthukrishnan2005data}; we give a short self-contained proof for our lower bound. 

\begin{theorem} \label{thm:eds}  \emph{($\bigstar$)}
 A single pass streaming kernelization algorithm for \EDSk requires at least $m-1$ bits of local memory for instances with $m$ edges.
\end{theorem}

\subsection*{General lower bound for a class of parameterized graph problems}\label{ssec:singlepass:generallb}

In the following we present space lower bounds for a number of parameterized graph problems. By generalizing the previous argument we find a common property that can be used to quickly rule out single pass kernels with $\Oh(poly(k) \log |x|)$ memory. We then provide a list of parameterized graph problems for which a single pass streaming kernelization algorithm requires at least $|E|- \Oh(1)$ bits of local memory.

\begin{definition}
 Let $Q \in \Sigma^* \times \mathbb{N}$ be a parameterized graph problem and let $c, k \in \N$. Then $Q$ has a $c$-$k$-\emph{stream obstructing graph} $G=(V, E)$ if $\forall e_i \in E$, there is a set of edges $R_i:= R(e_i) \subseteq \binom{V}{2} \setminus E$ of size $c$ such that $\forall F \subseteq E$, $(G[F\cup R_i], k) \in Q$ if and only if $e_i \in F$.
\end{definition}

In other words, each edge $e_i \in E$ could equally be critical to decide if $(G',k) \in Q$ for a graph instance $G'$ induced by a subset $F \subseteq E$ and a constant sized remainder of edges $R_i$, depending on what $R_i$ looks like. Note that $G$ may contain isolated vertices which can also be used to form edge sets $R_i$. We also consider $G$ to be a $c$-$k$-stream obstructing graph in the case that the above definition holds except that $\forall F \subseteq E$, $(G[F\cup R_i], k) \in Q$ if and only if $e_i \notin F$. We omit the proofs for this symmetrical definition in this section.

\begin{lemma} \label{lem:gen} \emph{($\bigstar$)}
 Let $Q \in \Sigma^* \times \mathbb{N}$ be a parameterized graph problem and let $c, k \in \N$. If $Q$ has a $c$-$k$-stream obstructing graph $G=(V,E)$ with $m$ edges, then a single pass streaming kernelization algorithm for $Q$ requires at least $m$ bits of local memory for instances with at most $m+c$ edges.
\end{lemma}

The following theorem is an easy consequence of Lemma~\ref{lem:gen} for problems that, essentially, have stream obstructing graphs for all numbers~$m$ of edges. Intuitively, of course also having such graphs only for an infinite subset of~$\N$ suffices to get a similar bound.

\begin{theorem} \label{thm:gen}
 Let $Q \in \Sigma^* \times \mathbb{N}$ be a parameterized graph problem. If there exist $c, k \in \mathbb{N}$ such that for every $m \in \mathbb{N}$, $Q$ has a $c$-$k$-stream obstructing graph $G$ with $m$ edges, then a single pass streaming kernelization algorithm for $Q$ requires at least $|E| - c$ bits of local memory.
\end{theorem}

\begin{proof}
 Let $A$ be a single pass streaming kernelization algorithm for $Q$. Assume that there is a stream obstructing graph $G_m=(V_m,E_m)$ for $Q$ with $m$ edges for every $m \in \mathbb{N}$. Then for every $m$ there is a group of instances $\mathcal{G}$ where for each $G_i=(V_i, E_i) \in \mathcal{G}$, $E_i = F \cup R_i$ for some $F \subseteq E_m$ and remainder of edges $R_i$ of size $c$, i.e. $|E_i| \leq m + c$. Let us consider all graph instances $G=(V, E)$ with exactly $|E| = m + c$ edges. Some of these
 instances are in $\mathcal{G}$, i.e. $E = E_m \cup R_i$ for some $R_i$. By Lemma \ref{lem:gen}, $A$ requires at least $m = |E| - c$ bits of local memory in order to distinguish these instances correctly. \qed \end{proof}
 
 The following corollary is a result of Theorem \ref{thm:gen} and constructions of stream obstructing graphs of arbitrary size for a variety of parameterized graph problems. We postpone these constructions to Appendix \ref{app:list}, where we will also exhibit proofs of correctness for a few of them.

\begin{corollary}
 For each of the following parameterized graph problems, a single pass streaming kernelization requires at least $|E| - \Oh(1)$ bits of local memory:
 \EDSk, \CEk, \CEDk, \CVDk, \CoVDk, \MFIk, \BERk, \FBVSk, \OCTk, \TEDk, \TVDk, \TPk, \sSPk, \BCNk.
\end{corollary}

\section{2-pass kernel for Edge Dominating Set} \label{sec:2pass}
Despite the previously shown lower bound of $m-1$ bits for a single pass kernel, there is a space efficient streaming kernelization algorithm for \textsc{Edge Dominating Set}($k$) if we allow it to make a pass over the input stream twice. We will first describe a single step of the streaming kernelization during the first pass. This is effectively a single pass kernel for finding a $2k$-vertex cover. After Step $t$ the algorithm has seen a set $A_t \subseteq E$. Some subset $A'_t \subseteq A_t$ of edges is stored in memory. Let us consider Step $t+1$ where a new edge $e = \{u, v\}$ arrives.

\smallskip

 1. Count the edges in $A'_{t}$ that are incident with $u$; do the same for~$v$.
 
 2. Let $A'_{t+1} = A'_t$ if either of these counts is at least $2k+1$.
 
 3. Otherwise, let $A'_{t+1} = A'_t \cup \{e\}$.
 
 4. If $|A'_{t+1}| > 4k^2+2k$, then return a \no instance.

\begin{lemma}\label{lem:vc} \emph{($\bigstar$)}
 After processing any set $A_t$ of edges on the first pass over the input stream the algorithm has a set $A'_t \subseteq A_t$ such
 that any set $S$ of at most $2k$ vertices is a vertex cover for $G[A_t]$ if and only if $S$ is a vertex cover for $G[A'_t]$.
\end{lemma}

Let $A'$ be the edges stored after the first pass. If there are more than $2k$ vertices with degree $2k+1$ in $G[A']$ then the algorithm returns a \no instance. We will continue with a description of a single step during the second pass. 
After Step $t$ the algorithm has revisited a set $B_t \subseteq E$. Some subset $B'_t \subseteq B_t$ of edges is stored along with
$A'$. Now, let us consider Step $t+1$ where the edge $e = \{u, v\}$ is seen for the second time.

\smallskip

 1. Let $B'_{t+1} = B'_t \cup \{e\}$ if $u, v \in V(A')$ and $e \notin A'$.
 
 2. Otherwise, let $B'_{t+1} = B'_t$.
 
 \smallskip

\noindent Let $B'$ be the edges stored during the second pass. The algorithm will return $G[A' \cup B']$, which is effectively $G[V(A')]$, after both passes have been 
processed without returning a \no instance.

\begin{lemma}\label{lem:eds} \emph{($\bigstar$)}
 After processing both passes the algorithm has a set $A' \cup B' \subseteq E$ such that there is an edge dominating set 
 $S$ of size at most $k$ for $G$ if and only if there is an edge dominating set $S'$ of size at most $k$ for $G[A' \cup B']$.
\end{lemma}

\begin{theorem} \label{thm:2pass} \emph{($\bigstar$)}
 \EDSk admits a two-pass streaming kernelization algorithm which, using $\Oh(k^3 \log n)$ bits of local memory and
 $\Oh(k^2)$ time in each step, returns an equivalent instance of size $\Oh(k^3 \log k)$.
\end{theorem}
  
If the algorithm stores a counter for the size of $A'_t$ and a tree structure $\mathbb{T}$ in which it maintains the number of sets (edges) in $A'_t$ that are a superset of $C \subseteq \binom{V}{\leq 2}$ as described in Section \ref{sec:singlepass}, then the operations in each step can be performed in $\Oh(\log k)$ time. We give the following corollary and omit the proof.

\begin{corollary}
 \EDSk admits a two-pass streaming kernelization algorithm which, using $\Oh(k^3 \log n)$ bits of local memory and $\Oh(\log k)$ time in each step, returns an equivalent instance of size $\Oh(k^3 \log k)$.
\end{corollary}

\section{Space lower bounds for multi-pass streaming kernels}\label{sec:multi}

In this section we will show lower bounds for multi-pass streaming kernels for \textsc{Cluster Editing}($k$) and \textsc{Minimum Fill-In}($k$). Similar to \textsc{Edge Dominating Set}($k$), it is difficult to return a trivial answer for these problems when the local memory exceeds a certain bound at some point during the input stream. Additional edges in the stream may turn a \no instance into a \yes instance and vice versa, which makes single pass streaming kernels infeasible. Although there is a 2-pass streaming kernel for \textsc{Edge Dominating Set}($k$), we will show that a $t$-pass streaming kernel for \textsc{Cluster Editing}($k$) requires at least $(n-2)/2t$ bits of local memory for instances with $n$ vertices. As a consequence, $\Omega(n)$ bits are required when a constant number of passes 
are allowed. Furthermore, $\Omega(n/\log n)$ passes are required when the streaming kernel uses at most $\Oh(\log n)$ bits of memory. We show a similar result for \textsc{Minimum Fill-In}($k$). 

\medskip
\noindent
\introduceparameterizedproblem{\CEk}{A graph $G = (V, E)$ and $k \in \N.$}{$k.$}{Can we add and/or delete at most~$k$ edges such that~$G$ becomes a disjoint union of cliques?}

Let us consider the following communication game with two players, $P_1$ and $P_2$. Let $N$ be a set of $n'$ vertices and let $u, v \notin N$. The players are given a subset of vertices, $V_1 \subseteq N$ and $V_2 \subseteq N$ respectively. Let $C(V_1)$ denote the edges of a clique on $V_1 \cup \{u\}$. Furthermore, let $S(V_2)$ denote the edges of a star with center vertex $v$ and leaves $V_2 \cup \{u\}$. The object of the game is for the players to determine if $G = (N\cup\{u, v\}, C(V_1)\cup S(V_2))$ is a disjoint union of cliques. The cost of the protocol for this game is the number of bits communicated between the players such that they can provide the answer. We can provide a lower bound for this cost by using the notion of fooling sets as shown in the following lemma.

\begin{lemma} \emph{(\cite{arora2009computational})} \label{lem:fool}
 A function $f: \{0,1\}^{n'} \times \{0, 1\}^{n'}$ has a size $M$ fooling set if there is an $M$-sized subset 
 $F \subseteq \{0, 1\}^{n'} \times \{0, 1\}^{n'}$ and value $b \in \{0, 1\}$ such that,
 
 \smallskip
 
 (1) for every pair $(x, y) \in S$, $f(x, y) = b$ 
  
 (2) for every distinct $(x, y), (x', y') \in F$, either $f(x, y') \neq b$ of $f(x', y) \neq b$.
 
 \smallskip

 \noindent If $f$ has a size-$M$ fooling set then $C(f) \geq \log M$ where $C(f)$ is the minimum number of bits communicated in a two-party protocol for $f$.
\end{lemma}

Let $f$ be a function modeling our communication game where $f(V_1, V_2) = 1$ if $G$ forms a 
disjoint union of cliques and $f(V_1, V_2) = 0$ otherwise. We provide a fooling set for $f$ in the following lemma.

\begin{lemma} \label{lem:foolce}
$f$ has a fooling set $F = \big\{(W, W) \mid W \subseteq N \big\}$.
\end{lemma}

\begin{proof}
For every $W \subseteq N$ we have $G = (N \cup \{u, v\}, C(W) \cup S(W))$ in which there is a clique on vertices $W \cup \{u, v\}$ while the vertices in $N \setminus W$ are completely isolated and thus form cliques of size 1, i.e. $f(W, W)=1$ for every $(W, W) \in F$. Now let us consider pairs $(W, W), (W', W') \in F$. We must show that either $f(W, W') = 0$ or $f(W', W) = 0$.
Clearly $W \neq W'$ since $(W, W) \neq (W', W')$. Let us assume w.l.o.g. that $W \setminus W' \neq \emptyset$, i.e. there is a vertex $w \in W \setminus W'$. Then $\{v, w\}, \{v, u\} \in S(W)$ since $w \in W$. However, $\{u, w\} \notin C(W')$ since $w \notin W'$ and by definition also $\{u, w\} \notin S(W)$ since $S(W)$ is a star with center $v \notin \{u,  w\}$. Thus, $G = (N \cup \{u, v\}, C(W') \cup S(W))$ is not a disjoint union of cliques, i.e., $f(W', W) = 0$ and the lemma holds. \qed
\end{proof}

The size of $F$ is $2^{n'}$, implying by Lemma \ref{lem:fool} that the protocol for $f$ needs at least $n'$ bits of communication. Intuitively, if we use less than $n'$ bits, then by the pigeonhole principle there must be some pairs $(W, W), (W', W') \in F$ for which the protocol is identical. Then the players cannot distinguish between the cases $(W, W), (W, W'), (W', W), (W', W')$, i.e. for each case the same answer will be given and thus the protocol is incorrect. We can now prove the following theorem by considering how the players could exploit knowledge of a multi-pass kernel for \CEk with small local memory in order to beat the lower bound of the communication game.

\begin{theorem} \label{thm:ce}
 A streaming kernelization algorithm for \CEk requires at least $(n-2)/2t$ bits of local memory for instances with $n$ vertices if it is allowed to make $t$ passes over the input stream.
\end{theorem}

\begin{proof}
 Let us assume that the players have access to a multi-pass streaming kernelization algorithm $A$ for \CEk.
 They can then use $A$ to solve the communication game for $|N| = n' = n-2$ by simulating passes over an input stream in the following way. First, $P_1$ initiates $A$ with budget $k=0$. To let $A$ make a pass over $C(V_1) \cup S(V_2)$, $P_1$ feeds $A$ with partial input stream $\langle C(V_1) \halt$. It then sends the current content of the local memory of $A$ to $P_2$, which is then able to resume $A$ and feeds it with $\resume S(V_2) \rangle$. In order to let $A$ make multiple passes, $P_2$ can send the local memory content back to $P_1$. Finally, when enough passes have been made an instance can be requested from $A$ for which the answer is \yes if and only if $f(V_1, V_2) = 1$.
 
 Now suppose $A$ is a $t$-pass streaming kernel with less than $(n-2)/2t$ bits of local memory for instances with $n$ vertices. In each pass the local memory is transmitted between $P_1$ and $P_2$ twice. Then in total the players communicate less than $n-2 = n'$ bits of memory. This is a contradiction to the consequence of Lemmata \ref{lem:fool} and \ref{lem:foolce}. Therefore $A$ requires at least $(n-2)/2t$ bits. \qed
\end{proof}

Note that this argument also holds for the \textsc{Cluster Deletion}($k$) and \textsc{Cluster Vertex Deletion}($k$)
problems where we are only allowed to delete $k$ edges, respectively vertices to obtain a disjoint union of cliques.

\medskip
\noindent
\introduceparameterizedproblem{\MFIk}{A graph $G = (V,E)$ and $k \in \N$.}{$k$.}{Can we add at most~$k$ edges such that~$G$ becomes chordal, i.e. $G$ does not contain an induced cycle of length 4?}

Let us consider the following communication game with two players, $P_1$ and $P_2$. Let $N$ be a set of $n$ vertices and let $p, u, v \notin N$. The players are given a subset of vertices, $V_1 \subseteq N$ and $V_2 \subseteq N$ respectively. Let $S_u(V_1)$ denote the edges of a star with center vertex $u$ and leaves $V_1 \cup \{p\}$. Furthermore, let $S_v(V_2)$ denote the edges of a star with center vertex $v$ and leaves $V_2 \cup \{p\}$. The object of the game is for the players to determine if $G = (N \cup \{p, u, v\}, S_u(V_1) \cup S_v(V_2))$ is a chordal graph. Let $f$ be a function modeling this communication game, i.e. $f(V_1, V_2) = 1$ if $G$ is chordal and $f(V_1, V_2) = 0$ otherwise. We provide a fooling set for $f$ in the following lemma.

\begin{lemma} \label{lem:foolmfi}
 $f$ has a fooling set $F = \big\{(W, N \setminus W) \mid W \subseteq N \big\}$.
\end{lemma}

The size of $F$ is $2^n$, implying by Lemma \ref{lem:fool} that the protocol for $f$ needs at least $n$ bits of communication. The following results from a similar argument to that of the proof of Theorem \ref{thm:ce}. We omit an explicit proof.

\begin{theorem}
 A streaming kernelization algorithm for \MFIk requires at least $(n-3)/2t$ bits of local memory for instances with $n$ vertices if it is allowed to make $t$ passes.
\end{theorem}

\section{Conclusion} \label{sec:conc}
In this paper we have explored kernelization in a data streaming model. Our positive results include single pass kernels for \dHSk and \dSMk, and a 2-pass kernel for \EDSk. We provide a tool that can be used to quickly identify a number of parameterized graph problems for which a single pass kernel requires $m-\Oh(1)$ bits of local memory for instances with $m$ edges. Furthermore, we have shown lower bounds for the space complexity of multi-pass kernels for \CEk and \MFIk.

\bibliographystyle{abbrv}
\bibliography{references}

\newpage
\appendix
\section{Proofs omitted from Section \ref{sec:singlepass}}
\subsection{Proof for Theorem \ref{thm:dhs}}

\begin{lemma}\label{lem:dhs}
After processing any set~$\F_t$ of edges on the input stream the algorithm has a set~$\F'_t\subseteq \F_t$ such that any set~$S$ of size at most~$k$ is a hitting set for~$\F_t$ if and only if~$S$ is a hitting set for~$\F'_t$.
\end{lemma}

\begin{proof}
We prove the lemma by induction. Clearly, the lemma is true for~$\F_0=\F'_0=\emptyset$. Now, assume that the lemma holds for all~$t\leq i$ and consider Step~$i+1$ in which, say, a set~$F$ appears on the stream.
Clearly, if there is a~$k$-hitting set~$S$ for~$\F_{i+1}$ then~$S$ is also a~$k$-hitting set for~$\F'_{i+1}\subseteq\F'_i\cup\{F\}\subseteq \F_i\cup\{F\}=\F_{i+1}$. I.e., this direction holds independently of whether the algorithm decides to store~$F$.

The converse, i.e., that a~$k$-hitting set for~$\F'_{i+1}$ is also a~$k$-hitting set for~$\F_{i+1}$, could only fail if the algorithm decided not to put~$F$ into~$\F'_{i+1}$; otherwise, such a~$k$-hitting set~$S$ would intersect~$F$ and all sets in~$\F'_i$, with the latter implying (by induction) that it intersects all sets in~$\F_i$. Assume that~$F\notin\F'_{i+1}$ which implies that the algorithm discovered a set~$C\subseteq F$ such that
\[
(d-|C|)!\cdot (k+1)^{d-|C|} 
\]
sets in~$\F'_i$ are supersets of~$C$. By the ordering of considered subsets~$C$ of~$F$ we know that for all~$C'\subseteq F$ of larger size there are strictly less than~$(d-|C'|)!\cdot (k+1)^{d-|C'|}$ sets containing~$C'$. Note that if~$C$ is contained in~$\F'_i$ then this already enforces that any hitting set for~$\F'_i$ also hits~$F$, so w.l.o.g.\ we assume that all sets are strict supersets of~$C$.

Let us consider the effect that adding~$F$ would have on~$\F'_i$, i.e., consider~$\hat{\F'_i}:=\F'_i\cup\{F\}$. By the previous considerations for~$\F'_i$ we conclude that in~$\hat{\F'_i}$ there are \emph{more than}~$(d-|C|)!\cdot (k+1)^{d-|C|}$ that contain~$C$ (since~$F$ also contains it). For all larger sets~$C'\subseteq F$ we reach a count of \emph{at most}~$(d-|C|)!\cdot (k+1)^{d-|C|}$. Crucially, this is where our invariant comes in, \emph{for all sets~$C'$ that are not subsets of~$F$} the counts are not increased when going from~$\F'_i$ to~$\hat{\F'_i}$, so there are also \emph{at most}~$(d-|C'|)!\cdot (k+1)^{d-|C'|}$ sets containing any~$C'\nsubseteq F$.

Now, for analysis, consider any maximal packing~$P=\{F_1,\ldots,F_\ell\}\subseteq \hat{\F'_i}$ of supersets of~$C$ such that the sets~$F_1\setminus C,\ldots,F_\ell\setminus C$ are pairwise disjoint (i.e., the sets pairwise overlap exactly in~$C$). This implies that all further supersets of~$C$ in~$\hat{\F'_i}$ must overlap some~$F_j\setminus C$. Let~$Q=\bigcup(F_j\setminus C)$ and note that the size of~$Q$ is at most~$\ell\cdot(d-|C|)$ with equality if all~$F_j$ have size~$d$. For any~$u\in Q$ we can consider~$C'=C\cup\{u\}$ and obtain that strictly less than
\[
(d-|C'|)!\cdot (k+1)^{d-|C'|}=(d-|C|-1)!\cdot (k+1)^{d-|C|-1}
\]
sets contain both~$C$ and~$u$. Since exactly~$(d-|C|)!\cdot (k+1)^{d-|C|}$ contain~$C$ as a strict subset (i.e., each contains at least one more element~$u'$), we get that
\begin{align*}
(d-|C|)!\cdot (k+1)^{d-|C|} < |Q|\cdot (d-|C|-1)!\cdot (k+1)^{d-|C|-1}
\end{align*}
which implies~$|Q|>(d-|C|)\cdot (k+1)$. Thus,~$\ell>k+1$, i.e.,~$\ell\geq k+2$. Now, we return from~$\hat{\F'_i}$ to~$\F'_i$ and note that even without having~$F\in\F'_i$ at least~$\ell-1\geq k+1$ of the sets~$F_1,\ldots,F_\ell$ are in~$\F'_i$. (We do not make any assumption about presence of~$F$ among these sets.) For ease of presentation let us rename~$k+1$ of those sets to~$F_1,\ldots,F_{k+1}\in\F'_i$.

Assume that~$\F'_i$ has a~$k$-hitting set, then by the induction hypothesis, there is also a~$k$-hitting set~$S$ for~$\F_i$. Since~$\F'_i\subseteq \F_i$ this set~$S$ must also be a hitting set for~$F_1,\ldots,F_\ell\in\F'_i\subseteq \F_i$. Since~$\ell\geq k+1$ some element~$s\in S$ must intersect at least two of the sets~$F_j,F_{j'}$, but then it also intersects the set~$C=F_j\cap F_{j'}$. Thus,~$S$ intersects also~$F\supseteq C$, implying that it is a~$k$-hitting set for~$\F'_{i+1}$ as claimed. This completes the inductive argument, and the proof.
\qed \end{proof}

Using the lemma, it is now straightforward to prove that the described algorithm is a streaming kernelization for \dHSk.

\begin{proof}
Correctness follows from Lemma \ref{lem:dhs}. As previously observed the algorithm stores at most $\Oh(k^d)$ sets at any
time during the computation. The elements of a set can be stored using $d \log |U|$ bits, i.e. $\Oh(k^d \log |U|)$ bits are used in total.
After the input stream has been processed the elements in $\F'_t$ can be relabeled such that they can be stored using $\Oh(\log k^d) = \Oh(\log k$) bits, i.e. an equivalent instance of size $\Oh(k^d \log k)$ bits is returned.
 
The algorithm iterates over at most $2^d$ subsets of the new set in each step. For each subset $C$ the number of sets in $\F'_t$ that are a superset of $C$ can be counted in $\Oh(|F'_t|) = \Oh(k^d)$ time, i.e. the algorithms spends $\Oh(2^d k^d) = \Oh(k^d)$ time in each step. \qed
\end{proof}

\subsection{Proof for Corollary \ref{cor:dhs}}

\begin{proof}
 We can find the corresponding node for a set $C$ by traversing $\mathbb{T}$ as follows. Let $n_0$ denote the root and let 
$C = \{c_1, \ldots, c_{|C|}\}$ such that $c_i < c_{i+1}$ for $1 \leq i < |C|$. Then node $n_{i+1}$ is assigned $c_{i+1}$ and is a 
child of node $n_i$. Finally, $n_{|C|}$ is the node corresponding to $C$ in which the number of sets in $\F'_t$ that contain a
superset of $C$ is stored. Each node in the $\mathbb{T}$ has at most $d \cdot |\F'_t|$ children. We can look up the child that is 
assigned $e_i$ in $\Oh(\log \F'_t)$ time by using the binary search tree. This step is performed $|C| \leq d$ times for each 
$C \subseteq F$, i.e. $\Oh(d 2^d  \log |F'_t|) = \Oh(\log k^d) = \Oh(\log k)$ time is spent in each step. The case that there is no node for $C$ in $\mathbb{T}$, i.e. there is no $F' \in \F'_t$ such that $C \subseteq F'$, can be identified similarly.

If the algorithm decides that $F \in \F'_{i+1}$ then the number stored in the node corresponding to $C$ is increased by 1 for 
each $C \subseteq F$ in increasing order of cardinality. In the case that there is no node for $C \subseteq F$ a new one will be
inserted at the appropriate place in $\mathbb{T}$, i.e. by finding the node in $\mathbb{T}$ corresponding to $C' = \{c_1, \ldots, c_{|C|-1}\}$ and inserting a new child for element $c_{|C|}$. Again, updating takes $\Oh(d \log F'_t)$ time for each $C \subseteq F$, i.e. $\Oh(\log k)$ time in each step.

Each set $F \in F'_t$ has at most $2^{|F|}$ subsets, each of which has a corresponding node in the tree. Then there are at most
$|\F'_t| \cdot 2^d$ nodes in the case that none of these subsets overlap. Instead of $|\F'_t| \cdot d \log |U|$, the algorithm now uses $|\F'_t| \cdot 2^d \log |U| = \Oh(k^d \log |U|)$ space. \qed
\end{proof}

\subsection{Proof for Theorem \ref{thm:dsm}}

\begin{lemma}
 After processing any set $\F_t$ of edges on the input stream the algorithm has a set $F'_t \subseteq \F_t$ such that there is a $k$-set matching $M$ for $\F_t$ if and only if there is a $k$-set matching $M'$ for $\F'_t$.
\end{lemma}

\begin{proof}
 We prove the lemma with an argument that is similar to the proof of Lemma \ref{lem:dhs} and point out the key differences. 
 Clearly, if there is a $k$-set matching $M$ for $\F'_{i+1}$ then $M$ is also a $k$-set matching for $\F_{i+1} = \F_i \cup \{F\} \supseteq \F'_i \cup \{F\} \supseteq \F'_{i+1}$.
 
 The converse, i.e., that a $k$-set matching $M$ for $\F_{i+1}$ implies the existence of a $k$-set matching $M'$ for $\F'_{i+1}$, could only fail if the algorithm decided not to put $F$ into $\F'_{i+1}$;
 if $F$ is not required for $M'$ then it certainly does not obstruct such a matching. 
 Then let us assume $F \notin \F'_{i+1}$ which implies that the algorithm discovered a set $C \subseteq F$ for which there are
 at least $d(k-1)+2$ supersets $F_1, \ldots, F_\ell$, such that their pairwise intersection is exactly $C$, and one of these sets is $F$.
 Then there are at least $d(k-1)+1$ such sets in $\F'_i$
 
 Assume that $F_{i+1}$ has a $k$-set matching $M$. If $F \notin M$, then $M$ is a matching for $\F_i$ and by induction hypothesis there is a $k$-set matching $M'$ for $\F'_{i} \subseteq \F_{i+1}$. In the case that $F \in M$, then there is at least one set among $F_1, \ldots, F_\ell$ that can take the role of $F$ in a matching $M'$. Each of these sets has at least one element that does not intersect $C$. Then a matching of size $k-1$ that does not contain a superset of $C$ can contain at most $d(k-1)$ of these sets. This leaves at least one set $F'$ that does not intersect any set in the $k-1$ matching. Then $F$ can be replaced by $F'$ in $M'$, i.e. there is a $k$-set matching in $\F'_i$ and therefore also in $\F'_{i+1}$. \qed \end{proof}
 
 The proof of Theorem \ref{thm:dsm} follows from the lemma in a similar way to the proof for \dHSk

\section{Proofs omitted from Section \ref{sec:singlepass:bounds}}
\subsection{Proof for Theorem \ref{thm:eds}}

\begin{proof}
Consider the following category of instances. Let $S \subseteq [n]$ and define ~$E(S)=\{\{a,s_i\}\mid i\in S\}$. Let $A$ be a single pass streaming kernelization algorithm for \EDSk and let $A$ receive budget $k=1$ and partial input stream $\langle \{a, s_{i_1}\}\ldots \{a, s_{i_{|S|}}\} \halt$,
where~$S=\{i_1,\ldots,i_{|S|}\}$, i.e., the stream contains exactly~$E(S)$ (the order therein is immaterial for our argument).
If $A$ uses less than $n$ bits of local memory, then by the pigeonhole principle there must be a pair $S', S'' \subseteq [n]$ such that $S' \neq S''$ where $\langle S' \halt$ and $\langle S'' \halt$ result in the same memory state. 

Now $A$ must return the same problem kernel for $\langle S', e \rangle$ and $\langle S'', e \rangle$ for every edge $e$ since its behavior always depends on its memory and the rest of the input stream. Let us assume w.l.o.g.\ that $S' \setminus S'' \neq \emptyset$ and let~$i\in S'\setminus S''$. It follows that~$\{a,s_i\}$ is an edge dominating set for the instance with edge set~$E(S')\cup\{\{b,s_i\}\}$, making this instance \yes. The instance with edge~$E(S'')\cup\{\{b,s_i\}\}$, however, has two connected components and thus is \no for budget~$k=1$. Thus,~$A$ cannot answer correctly for both instances; contradiction. Thus, any streaming kernelization must use at least~$n$ bits for this type of instance with~$n+1$ edges.

For every $m \in \mathbb{N}$, if we set $n = m-1$ then there is an instance with $m$ edges for which $A$ requires at least $m-1$ bits. Therefore, any single pass streaming kernelization algorithm for \EDSk requires at least $m-1$ bits of local memory. \qed \end{proof}
\subsection{Proof for Lemma \ref{lem:gen}}
\begin{proof}
 Let $A$ be a single pass streaming kernelization algorithm for $Q$ using less than $m$ bits. We consider a worst case scenario for the ordering in which edges appear in the input stream. If a streaming kernel requires some minimum amount of memory for this ordering, then it requires at least as much memory when the edges appear in some arbitrary order. 
 
 Let us consider instances with at most $m+c$ edges that have an input stream of the type $\langle F, R \rangle$. 
 That is, a subset of edges $F \subseteq E$ appears first, followed by some set $R$ of size $c$. If $A$ uses less than $m$ bits of local memory, then by the pigeonhole principle there must be a pair of subsets $F', F'' \subseteq E$ such that $F' \neq F''$, where $\langle F' \halt$ and $\langle F'' \halt$ result in the same memory state.
 
 Now $A$ must return the same problem kernel for $\langle F', R \rangle$ and $\langle F'', R \rangle$ for every $\resume R \rangle$. Let us assume w.l.o.g. that $F' \setminus F'' \neq \emptyset$ and let~$e_i \in F' \setminus F''$. Thus, for the corresponding set $R_i=R(e_i)$ we have $(G[F' \cup R_i], k) \in Q$ and $(G[F'' \cup R_i],k) \notin Q$. We conclude that $A$ is not a correct kernelization algorithm since it cannot answer both instances correctly if being in the same state after $\langle F' \halt$ and $\langle F'' \halt$. Therefore, any single pass streaming kernelization algorithm for $Q$ requires at least $m$ bits of local memory for instances with at most $m+c$ edges.
\qed \end{proof}

\section{Construction of Stream Obstructing Graphs} \label{app:list}
\paragraph{Edge Dominating Set.} The construction for a 1-1-stream obstructing graph for \EDSk with an arbitrary number of edges is implicitly used in the proof of Theorem \ref{thm:eds}.

\paragraph{Cluster Editing.} We construct a 1-0-stream obstructing graph $G = (V, E)$ for \CEk with $m$ edges in the following way.
Let $E$ be a set of pairwise disjoint edges such that $|E| = m$ and let $V = V(E) \cup \{w\}$ where $w \notin V(E)$. The following lemma shows that this construction suffices.

\begin{lemma}
$G$ is a 1-0-stream obstructing graph for \CEk. 
\end{lemma}

\begin{proof}
For every edge $e = \{u, v\} \in E$, choose $R(e) = \{\{v, w\}\}$, i.e. $|R(e)| = c = 1$. For every subset $F \subseteq E$ we show that $G[F \cup R(e)]$ is a cluster graph if and only if $e \notin F$ since we have a budget of 0. Suppose that $e \in F$. Then $G[F \cup R(e)]$ has an induced $P_3$ on vertices $u, v, w$, i.e. it is not a cluster graph. In the other case suppose that $e \notin F$. Then $G[F \cup R(e)]$ is a set of $|F|+1$ disjoint edges since $\{v, w\}$ does not intersect with any edge in $F$, i.e. it is a cluster graph. This completes the proof. \qed
\end{proof}

We observe that $G$ is also a 1-0-stream obstructing graph for \CEDk and \CVDk since a budget of 0 forces any graph in a \yes instance to be a cluster graph.

If we choose $R(e) = \{\{u, w\}, \{v, w\}\}$, then we can show that $G$ is a 2-0-stream obstructing graph for \BERk, \FBVSk, \OCTk, \TEDk and \TVDk since any induced triangle is a forbidden structure in any instance with budget 0 for these problems. There is a triangle in instances $G[F\cup R(e)]$ if and only if $e \in F$, namely on vertices $u, v, w$. Furthermore, $G$ is a 2-1-stream obstructing graph for \TPk since there is a single triangle in the graph if and only if $e \in F$.

\paragraph{Cograph Vertex Deletion.} We construct a 2-0-stream obstructing graph $G = (V,E)$ for \CoVDk with $m$ edges in the following way. Let $E$ be a set of pairwise disjoint edges such that $|E| = m$ and let $V = V(E) \cup \{p, w\}$ where $p, w \notin V(E)$. For every $e = \{u, v\} \in E$ we choose $R(e) = \{\{p, u\},\{v,w\}\}$, i.e. $|R(e)| = c = 2$. Similar to the stream obstructing graph for \CEk, a budget of 0 forces any graph in a \yes instance to be a cograph. This only holds for graphs $G[F \cup R(e)]$ for subsets $F \subseteq E$ if $e \notin F$ since otherwise the graph has an induced $P_4$ on vertices $p, u, v, w$.

If we choose $R(e) = \{\{p, u\}, \{v, w\}, \{p, w\}\}$, then we can show that $G$ is a 3-0-stream obstructing graph $G = (V, E)$ for \MFIk since there is a $C_4$ in instances $G[F \cup R(e)]$ if and only if $e \in F$, namely on vertices $p, u, v, w$.

\paragraph{s-Star Packing.}
We construct a $(s-1)$-1-stream obstructing graph $G = (V,E)$ for \sSPk with $m$ edges in the following way. Let $E$ be a set of pairwise disjoint edges such that $|E| = m$ and let $V = V(E) \cup \{w_1, \ldots, w_{s-1}\}$ where $w_1, \ldots w_{s-1} \notin V(E)$. We show that this construction suffices.

\begin{lemma}
 $G$ is a $(s-1)$-1-stream obstructing graph for \sSPk.
\end{lemma}

\begin{proof}
For every edge $e = \{u, v\} \in E$, choose $R(e) = \{\{v,w_1\},\ldots\{v,w_{s-1}\}\}$, i.e. $|R(e)| = c = s-1$. For every subset $F \subseteq E$ we show that $G[F\cup R(e)]$ contains exactly one instance of $K_{1,s}$. Suppose that $e \in F$. Then $G[F\cup R(e)]$ contains a $K_{1,s}$ with center vertex $v$ and leaves $u, w_1, \ldots, w_{s-1}$. In the other case suppose that $e \notin F$. Then $G[F\cup R(e)]$ is a set of disjoint edges plus a $K_{1,s-1}$ on center vertex $v$ with leaves $w_1, \ldots, w_{s-1}$. \qed
\end{proof}

\paragraph{Bipartite Colorful Neighborhood.}
In a natural extension to graph streaming we assume that edges $e = \{u, w\}$ in a stream for bipartite graphs $G(U \cup W,E)$ are given such that $u \in U$ and $w \in W$. We construct a 1-1-stream obstructing graph $G = (U \cup W,E)$ for \BCNk with $m$ edges in the following way. Let $E$ be a set of pairwise disjoint edges such that $|E| = m$ and let $U \cup W = V(E) \cup \{v\}$ where $v \notin V(E)$ and $v \in W$. The following lemma shows that this construction suffices.

\begin{lemma}
 $G$ is a 1-1-stream obstructing graph for \BCNk.
\end{lemma}

\begin{proof}
 For every edge $e = \{u, w\} \in E$, choose $R(e) = \{\{u, v\}\}$, i.e. $|R(e)| = c = 1$. For every subset $F \subseteq E$ we show that $G[F \cup R(e)]$ has a vertex in $U$ with two neighbors in $W$ if and only if $e \in F$. Suppose that $e \in F$. Then $G[F \cup R(e)]$ has $u \in U$ and edges $\{u, v\}$ and $\{u, w\}$, i.e. there is a vertex $u$ with a neighborhood in $W$ that can be two-colored. In the other case, suppose that $e \notin F$. Then vertices in $U$ in the graph $G[F\cup R(e)]$ have at most one neighbor.
 \qed
\end{proof}

\section{Proofs omitted from Section \ref{sec:2pass}}

\subsection{Proof for Lemma \ref{lem:vc}}

\begin{proof}
 We prove the lemma by induction. Clearly, the lemma is true for $A_0 = A'_0 = \emptyset$. Now, assume that the lemma holds for all $t \leq i$ and consider Step $i+1$ in which an edge $e$ appears on the stream. First, suppose that $A'_{i+1} > 4k^2+2k$. Vertices are incident to at most $2k+1$ edges in $A'_{i+1}$, i.e. $2k$ vertices can cover at most $2k(2k+1) = 4k^2+2k$ edges in $A'_{i+1}$. Therefore, there can be no solution and the algorithm can safely return a \no instance. 
 
 In the other case, let us assume that $|A'_{i+1}| \leq 4k^2+2k$. Clearly, if there is a $2k$-vertex cover $S$ for $A_{i+1}$, then $S$ is also a $2k$-vertex cover for $A'_{i+1} \subseteq A'_i \cup \{e\} \subseteq A_i \cup \{e\} = A_{i+1}$. 
 The converse, i.e. that a $2k$-vertex cover for $A'_{i+1}$ is also a $2k$-vertex cover for $A_{i+1}$ could
 only fail if the algorithm decided not to put $e$ in $A'_{i+1}$; otherwise, such a $2k$-vertex cover $S$ would cover $e$ and all edges in $A'_i$, with the latter implying (by induction) that $S$ covers all edges in $A_{i+1} = A_i \cup \{e\}$.
 
 Then let us assume $e \notin A'_{i+1}$ which implies that the algorithm discovered a vertex $v$ that is incident to $e$ and at
 least $2k+1$ edges in $A'_i$. Then any vertex cover $S$ of size at most $2k$ must contain $v$ in order to cover these edges, i.e.
 $S$ will cover $e$ in any case. Thus $S$ is also a $2k$-vertex cover for $A_{i+1}$. \qed \end{proof}
 
\subsection{Proof for Lemma \ref{lem:eds}}

\begin{proof}
 Suppose $S$ is an edge dominating set of size at most $k$ for $G[A' \cup B']$. Then $V(S)$ is a vertex cover of size at most $2k$  for $G[A' \cup B']$ and therefore also for $G[A']$. By Lemma \ref{lem:vc} we have that $V(S)$ is also a vertex cover for $G$. Therefore $S$ is also an edge dominating set for $G$ since the endpoints of edges in $S$ cover all edges in $G$ and $S \subseteq A' \cup B' \subseteq E$.
 
 For the converse, suppose that $S$ is an edge dominating set of size at most $k$ for $G$. Then $V(S)$ is a vertex cover of size at most $2k$ for $G[A']$. Each edge in $B'$ is incident with at least one vertex $v$ such that $2k+1$ edges in $A'$ are incident with $v$, i.e. $v$ must be part of a $V(S)$. Therefore $V(S)$ is also a vertex cover for $G[A' \cup B']$ since it also covers all edges in $B'$.
 
 Now let us verify that there is an edge dominating set $S'$ of size at most $k$ for $G[A' \cup B']$. We will show how to find $S'$ by considering edges $e = \{u, v\}$ of $S$. First, let $S' = \emptyset$. If $u, v \in V(A')$, then $e \in A' \cup B'$, i.e. add $e$ to $S'$ in order to cover neighboring edges of $e$.
 If $u, v \notin V(A')$, then $S'$ does not require $e$ since in this case every edge in $A' \cup B'$ is incident with neither $u$ nor $v$.
 In the remaining case we have w.l.o.g. $u \in V(A')$, $v \notin V(A')$, i.e. there are no edges in $A' \cup B'$ that
 are incident with $v$. Then $e$ can be substituted by any other edge $e'$ in $A' \cup B'$ that is incident with $u$, i.e. add $e'$ to $S'$. At the end $S'$ is an edge dominating set of size at most $k$ for $G[A' \cup B']$. \qed \end{proof}

\subsection{Proof for Theorem \ref{thm:2pass}}

\begin{proof}
 Correctness of the algorithm follows from Lemma \ref{lem:eds}. After the first pass there is a set $A' \subseteq E$ with at most
 $2k(2k+1)$ edges. Let $H$ be the set of vertices in $V(A')$ of degree at least $2k+1$ and let $L$ be the set of vertices of degree at most $2k$, i.e. $|H| = \Oh(k)$ and $|L| = \Oh(k^2)$.
 
 After the second pass there are $\Oh(k^2)$ edges in $A' \cup B'$ that are incident with two vertices in $H$. None of the edges that are incident with two vertices in $L$ were discarded in the first pass, i.e. there are $\Oh(k^2)$ such edges in $A' \cup B'$. Finally, there are at most $\Oh(k^3)$ edges that are incident with a vertex in $H$ and a vertex in $L$ since every vertex in $L$ has at most $2k$ neighbors in $H$.
 
 An edge can be stored using $\Oh(\log n)$ bits and  $|A' \cup B'| = \Oh(k^3)$. The algorithm stores a subset of $A' \cup B'$ at any time during the execution which requires $\Oh(k^3 \log n)$ bits. Therefore, the algorithm uses $\Oh(k^3 \log n)$ bits of memory in each step. After both passes have been processed the vertices of the equivalent instance can be relabeled such that they can be stored using $\Oh(\log k)$ bits, i.e. the size of the instance is now $\Oh(k^3 \log k)$. 
 
 Counting the size of $A'_t$ and the number of edges in $A'_t$ that are incident with a certain vertex $v$ can be performed in $\Oh(A'_t) = \Oh(k^2)$ time. Similarly, verifying if an edge is in $A'_t$ and verifying if an edge in $A'_t$ is incident with a certain vertex $v$ can be done in $\Oh(k^2)$ time.
  \qed \end{proof}
  
\section{Proofs omitted from Section \ref{sec:multi}}
\subsection{Proof for Lemma \ref{lem:foolmfi}}

\begin{proof}
For every $W \subseteq N$ we have $G = (N \cup \{p, u, v\}, S_u(W) \cup S_v(N\setminus W))$ which is cycle free since $S_u(W) \cup S_v(N \setminus W)$ forms a tree rooted at $p$ with leaves $N$. Therefore $G$ is chordal, i.e. $f(W, N \setminus W) = 1$ for every $(W, N \setminus W) \in F$. Now let us consider pairs $(W, N \setminus W), (W' \setminus W') \in F$. We must show that either $f(W, N \setminus W') = 0$ or $f(W', N \setminus W) = 0$. Clearly $W \neq W'$ since $(W, N \setminus W) \neq (W', N \setminus W')$. Let us assume w.l.o.g. that $W \setminus W' \neq \emptyset$, i.e. there is a vertex $w \in W \setminus W'$. Then $\{u, w\}, \{u, p\} \in S_u(W)$ since $w \in W$. Furthermore $\{v, w\}, \{v, p\} \in S_v(N\setminus W')$ since $w \in N \setminus W'$ because $w \notin W'$. Thus $G = (N \cup \{p, u, v\}, S_u(W) \cup S_v(W'))$ contains an induced cycle on 4 vertices and is not chordal, i.e., $f(W, N \setminus W') = 0$ and the lemma holds.
\end{proof}

\end{document}